\documentclass[a4paper,UKenglish,cleveref,autoref,thm-restate]{lipics-v2021}

\usepackage{mathtools}

\title{Eventual Nonnegativity of a Matrix Is in P}
\titlerunning{Eventual Nonnegativity of a Matrix Is in P}

\author{Julian D'Costa}{Independent, London, UK}{julianrdcosta@gmail.com}{https://orcid.org/0000-0003-2610-5241}{}
\authorrunning{J.~D'Costa}
\Copyright{Julian D'Costa}



\supplement{The structural criterion and key algebraic inequalities have been formalised in Lean 4 and machine-checked.}
\supplementdetails[subcategory={Lean 4 formalisation}]{Software}{https://github.com/julianrdcosta/ptime-ENN-lean-formalization}

\nolinenumbers

\hideLIPIcs

\EventEditors{John Q. Open and Joan R. Access}
\EventNoEds{2}
\EventLongTitle{42nd Conference on Very Important Topics (CVIT 2016)}
\EventShortTitle{CVIT 2016}
\EventAcronym{CVIT}
\EventYear{2016}
\EventDate{December 24--27, 2016}
\EventLocation{Little Whinging, United Kingdom}
\EventLogo{}
\SeriesVolume{42}
\ArticleNo{23}

\newcommand{\Q}{\mathbb{Q}}
\newcommand{\N}{\mathbb{N}}
\newcommand{\Z}{\mathbb{Z}}
\newcommand{\Gal}{\operatorname{Gal}}
\newcommand{\poly}{\operatorname{poly}}

\begin{document}

\maketitle

\begin{abstract}
Given a rational matrix $A$, is every sufficiently large power $A^n$
entrywise nonnegative?  We prove that this problem is decidable in
deterministic polynomial time.  This is in contrast with deciding eventual
nonnegativity of a single prescribed entry sequence $(A^n)_{ij}$, which
amounts to the Ultimate Positivity Problem for linear recurrence
sequences, whose decidability is open.

The previous decidability procedure~\cite{dcosta-ow} splits the matrix powers into residue
classes modulo a torsion exponent $D$ whose value can be exponential in
the input size.  We show that it is sufficient to check a single progression $A^{Dk+1}$:
the exponents $n$ with $A^n\geq0$ are closed under addition, and two
consecutive exponents of the progression are coprime, so they generate
every sufficiently large exponent. 

Along the progression, eigenvalues are grouped by their common $D$th
power.  The coefficient of each group is computed without forming
$A^D$, any $\lambda^D$, or a splitting field: each summand is evaluated
in its own root field $\Q(\lambda)$, and Galois equivariance keeps the
degree and height of every class sum polynomially bounded.  This is
enough for certified zero and sign tests.  The same
algorithm with one case rejected decides eventual positivity, and a
variant decides whether the matrix has any nonnegative power at all.

\end{abstract}

\section{Introduction}\label{sec:intro}

A real square matrix $A$ is \emph{eventually nonnegative} if there is an
integer $N$ such that $A^n\geq0$ entrywise for every $n\geq N$.
The matrix itself may contain negative entries, and its early powers may have entries that 
change sign repeatedly; eventual nonnegativity asks whether this transient
behaviour ultimately disappears forever.

A natural question that arises is whether we can determine whether a given matrix is eventually nonnegative. This question is closely related to the Ultimate Positivity Problem for linear recurrence sequences, as we now discuss.

\subsection{From one entry to the whole matrix}
Fix a rational $d\times d$ matrix $A$ and indices $i,j$.  The
sequence of entries $(A^n)_{ij}$ is a rational linear recurrence
sequence (LRS) of order at most $d$: by the Cayley--Hamilton theorem,
$A^{n+d}$ is a fixed rational linear combination of
$A^{n+d-1},\ldots,A^{n}$, and reading off the $(i,j)$ entry of this
identity shows that $(A^n)_{ij}$ satisfies the corresponding order-$d$
recurrence.

Conversely, every rational LRS appears, up to a shift of the index, as
an entry of the powers of some rational matrix. Deciding eventual
nonnegativity of one
prescribed entry sequence is therefore a formulation of the Ultimate
Positivity Problem for linear recurrence sequences, whose decidability remains
open~\cite{akshay-weighted,dcosta-ow,ow-ultimate}.  Eventual
nonnegativity appears to demand the conjunction of $d^2$ such properties.
Nevertheless, we prove:
\begin{theorem}\label{thm:main}
Eventual nonnegativity of rational matrices is decidable in deterministic
polynomial time.
\end{theorem}
This is possible because the $d^2$ entry sequences are not independent: they come
from a single multiplicative orbit, and a yes-instance has rigid algebraic
structure that an arbitrary scalar LRS lacks. 

In the other direction, eventual nonnegativity is more delicate
than eventual \emph{positivity}, in which all entries must become strictly
positive.  Eventual positivity admits a characterization of
Perron--Frobenius type~\cite{noutsos}, checkable in polynomial time: a
single positive eigenvalue dominates the picture---more precisely, a
simple eigenvalue of strictly maximal modulus whose left and right
eigenvectors are entrywise positive.  Allowing zeros
destroys that rigidity.  The matrix may be reducible, may have nilpotent components and
nontrivial Jordan blocks; several eigenvalues of maximum modulus may
interact; and their spectral contributions may cancel exactly.  

\subsection{Why the previous procedure is not polynomial}
D'Costa, Ouaknine, and Worrell gave an effective characterization of
eventually nonnegative rational matrices~\cite{dcosta-ow}, answering a
question of Akshay, Chakraborty, and Pal~\cite{akshay-weighted}.  Their
argument chooses an exponent $D$ that kills every root-of-unity quotient
of distinct eigenvalues---so that $\lambda^D=\gamma^D$ whenever
$\lambda/\gamma$ is a root of unity---and analyses the subsequences
$(A^{Dk+r})_{k\geq0}$ for each residue $0\leq r<D$.  Each such
subsequence is \emph{nondegenerate}: in its closed form, no quotient of
two distinct exponential bases is a root of unity.  This proves
decidability, but although $D$ has polynomially many bits, its value can
be exponential, so enumerating all $D$ residues is exponential.  

A second difficulty is independent of the enumeration: the subsequences
cannot be handled through their explicit closed forms either.  Their
exponential bases are the $D$th powers $\lambda^D$, and distinct
eigenvalues can coalesce, since $\lambda^D=\gamma^D$ as soon as
$\lambda/\gamma$ is a root of unity killed by $D$.  Forming $A^D$ or any $\lambda^D$ may produce exponentially
many bits, and adjoining all eigenvalues to a common field may produce a
splitting field of degree $d!$.  A polynomial-time algorithm must therefore decide exact
cancellation among such cross-field contributions \emph{without}
constructing any of these objects.

\subsection{Three ideas}
The proof of Theorem~\ref{thm:main} is organized around the following three ideas.
\begin{enumerate}
\item \emph{One residue class suffices.}
      The set $S=\{n\geq1:A^n\geq0\}$ is closed under addition, because a
      product of nonnegative matrices is nonnegative.  An additively
      closed subset of $\N$ containing two coprime elements contains
      every sufficiently large integer, by the Frobenius coin theorem.  Hogben and Wilson package
      this combination for matrices~\cite{hogben-wilson}: they call a
      matrix property \emph{power-hereditary} when it passes from $M^k$
      and $M^\ell$ to $M^{k+\ell}$, and they prove that such a property
      holding at finitely many exponents with greatest common divisor $1$
      holds at all sufficiently large
      exponents~\cite[Corollary~2.5]{hogben-wilson}.
      The progression itself is also classical: a torsion-killing
      modulus with offset $1$ appears, for spectral rather than
      algorithmic ends, inside Zaslavsky and Tam's proof of Friedland's
      lemma~\cite[proof of Theorem~3.1]{zaslavsky-tam}, reproduced
      as~\cite[Theorem~4.8]{hogben-wilson}.  What is new here is its
      algorithmic role: collapsing the enumeration over residues.

      We choose the
      exponents so that the Hogben--Wilson principle
      applies to specific powers where the spectral analysis is tractable. Let $D$ be an even
      exponent killing all root-of-unity quotients of distinct nonzero
      eigenvalues. Notice that consecutive exponents of the progression $Dk+1$ are
      coprime, therefore
      \[
       \text{$A$ is eventually nonnegative}
       \quad\Longleftrightarrow\quad
       \text{$A^{Dk+1}\geq0$ for all sufficiently large $k$}.      
      \]
      The quantification over all $D$ residue classes thus
      reduces to checking a single well-chosen residue class.
\item \emph{Matrix rigidity.}
      Along the chosen progression the spectral picture simplifies.
      Every entry sequence $(A^{Dk+1})_{ij}$ is an LRS in $k$: a sum of
      terms $P(k)\,\nu^k$ with polynomial coefficients, whose
      bases $\nu$ are the $D$th powers $\lambda^D$ of eigenvalues.  By the
      choice of $D$, no quotient of two distinct bases is a
      root of unity, so the sequences are nondegenerate.
      Suppose $A^{Dk+1}\geq0$ for all $k\geq k_0$, and sample one entry
      on a further arithmetic subsequence: with $N_0=Dk_0+1$,
      \[
       v_m=\bigl(A^{N_0(1+Dm)}\bigr)_{ij}
        =e_i^\top A^{N_0}\bigl(A^{N_0D}\bigr)^m e_j .
      \]
      Since $A^{N_0D}=(A^{N_0})^D\geq0$, the sequence $v_m$ has the form
      $x^\top M^my$ with $x$, $M$, $y$ entrywise nonnegative.  A theorem
      of Berstel on the poles of nonnegative rational
      series~\cite{berstel-1971,berstel-reutenauer} says that a
      nondegenerate LRS with such a nonnegative realization has a
      \emph{unique} exponential basis of maximum modulus.  As discussed in~\cite{dcosta-ow}, tracing back from the subsequence $v_m$ to the progression $A^{Dk+1}$,
      every entry that is not eventually zero has a unique dominant
      contribution, and its coefficient must be positive.
\item \emph{Galois equivariance lets us avoid the splitting field.}
      The resulting criterion must still be checked efficiently.  The
      coefficient of a base $\mu=\lambda^D$ in the closed form of
      $(A^{Dk+1})_{ij}$ is a \emph{grouped coefficient}: a sum, over all
      eigenvalues $\lambda$ with $\lambda^D=\mu$ (a \emph{torsion
      class}), of individual spectral contributions, each computable
      inside its own root field $\Q(\lambda)$ of degree at most $d$.  The
      summands of a class need not share a small common field, so the
      naive degree bounds for their sum are exponential: adjoining all
      terms could require their compositum, and the splitting field of
      the characteristic polynomial $\chi_A$ can have degree
      $d!$.
      The remedy is Galois symmetry, the usual way of avoiding a
      splitting field; what is particular to our setting is the structure
      that the symmetry preserves, namely the torsion partition.  Galois
      conjugation permutes conjugate eigenvalues and respects torsion
      relations, and conjugating the closed form of the rational
      sequence $(A^{Dk+1})_{ij}$ term by term shows that every conjugate
      of a grouped coefficient is again the sum over a complete torsion
      class---one of at most $d$ candidates.
      
      A grouped coefficient therefore has
      algebraic degree at most $d$, and separately, direct estimates give it
      polynomial logarithmic height.  A nonzero grouped coefficient is
      then separated from zero by an explicit $2^{-\poly}$
      Liouville-type bound, and certified ball
      arithmetic decides exact
      zero and sign in polynomial time.
\end{enumerate}

The result is a decision algorithm for eventual nonnegativity: it determines
whether there is a threshold beyond which the matrix powers are nonnegative, without computing the least one; on a
yes-instance a valid threshold with polynomially many bits can nevertheless
be certified (Remark~\ref{rem:threshold}). 

The same algorithm with one case
rejected decides eventual positivity in polynomial time, independently
of~\cite{noutsos}. A variant of the analysis further shows that the
existence of a nonnegative power of a given matrix---eventual or not---is
decidable in polynomial time (Corollary~\ref{cor:existspower}).

\subsection{Related work}
Eventual nonnegativity was introduced by Friedland in connection with the
nonnegative inverse eigenvalue problem~\cite{friedland}, and its Jordan
and combinatorial structure was developed by Naqvi and
McDonald, among others~\cite{naqvi-mcdonald-comb,naqvi-mcdonald}.  Hogben and Wilson's
general theory of eventual matrix properties supplies the
power-hereditary principle used above~\cite[Corollary~2.5]{hogben-wilson};
they also answer a question of Zaslavsky and Tam~\cite{zaslavsky-tam} in the same spirit:
nonnegativity at an exponent $\equiv1\pmod p$ for every modulus $p$
implies eventual nonnegativity~\cite[Corollary~2.10]{hogben-wilson}.
Hogben gives a spectral test for the stronger notion of strong eventual
nonnegativity~\cite{hogben-strong}.
Matrices possessing at least one nonnegative power---\emph{power
nonnegative} matrices---and their relationship to eventual
nonnegativity are studied by Tudisco, Cardinali, and
Di Fiore~\cite{tudisco-cdf}; Corollary~\ref{cor:existspower} shows that
membership in this class is decidable in polynomial time for rational
matrices.

Eventual positivity also arises in control theory.  There, a
discrete-time linear \emph{system} is a recurrence
$x_{k+1}=Ax_k+Bu_k$, $y_k=Cx_k$, transforming an input sequence
$(u_k)$ into an output sequence $(y_k)$ through a hidden internal
state $x_k$.  The matrix triple $(A,B,C)$ is a \emph{realization} of
the resulting input--output map, and it is \emph{minimal} if no
realization with a smaller state dimension induces the same map.  The system is
\emph{externally positive} if nonnegative inputs always produce
nonnegative outputs; equivalently, its impulse response $CA^nB$,
$n\geq0$, is entrywise nonnegative.  Altafini~\cite{altafini} showed
that an externally positive system may admit no minimal realization
in which $A$, $B$, $C$ are all nonnegative, yet admit one whose state
matrix $A$ is eventually positive---the negativity is confined to a
transient of the internal dynamics.  Deciding eventual positivity of
a given state matrix is exactly the problem studied here
(Remark~\ref{rem:positivity}), and Corollary~\ref{cor:existspower}
decides whether the state matrix becomes nonnegative under some
\emph{decimation}, that is, after re-sampling the dynamics every
$n$th step, which replaces $A$ by $A^n$.  In infinite dimensions,
eventual positivity of operator semigroups is an active line of
research in functional analysis~\cite{daners-glueck-kennedy}.

On the LRS side, Ultimate Positivity
is decidable in PSPACE for \emph{simple} LRS (those without repeated
characteristic roots), and in polynomial time for fixed
order~\cite{ow-ultimate}.  In the weighted setting---a sum of powers
of several matrices with scalar weights---already two summands recover
the difficulty of Ultimate Positivity~\cite{akshay-weighted}.  

The obstacle behind Section~\ref{sec:grouped} is a familiar one in
computer algebra.  The splitting field of a polynomial of degree $d$
can have degree $d!$, so an algorithm that works inside it is
exponential, whereas each root field $\Q(\lambda)$ has degree at most
$d$.  Exact linear algebra within a single root field goes back to
Cai's Jordan algorithm~\cite{cai} and rests on polynomial factorization
over $\Q$~\cite{lll,yap}; the difficulty lies in combining quantities
that live in different root fields.  Combot~\cite{combot} recently
overcame it for a classical question---which linear and multiplicative
relations hold among the roots of a polynomial---by showing that every
such relation can be detected inside the root fields, so that the
lattice of relations is computable in polynomial time.  Our need is
more modest.  A grouped coefficient is a sum of contributions from
different root fields, and we must decide whether it is zero and, when
it is real, its sign.  Because the torsion partition is
Galois-invariant, each such sum has degree at most the number of
classes, and with the height bounds of Appendix~\ref{app:heights} this
is enough for a certified test.

\subsection{Comment on the use of AI}
The results in this paper were obtained by the OpenAI model GPT-5.6 Sol on 2nd August 2026. The author takes responsibility for correctness, originality, and exposition. A sorry-free Lean formalisation of the criterion and the key inequalities behind the complexity analysis is available at \url{https://github.com/julianrdcosta/ptime-ENN-lean-formalization}. Verifying an end-to-end polynomial-time algorithm was unfortunately infeasible since Mathlib as of writing lacks infrastructure for proving algorithmic complexity bounds.

\section{One residue class suffices}\label{sec:one-residue}

\begin{lemma}[Coprime-residue lemma]\label{lem:coprime-residue}
For every integer $D\geq1$ and every real square matrix $A$,
\[
 A\text{ is eventually nonnegative}
 \quad\Longleftrightarrow\quad
 A^{Dk+1}\geq0\text{ for every sufficiently large }k.
\]
\end{lemma}

This is an immediate specialization of the power-hereditary principle of
Hogben and Wilson~\cite[Corollary~2.5]{hogben-wilson};
entrywise nonnegativity is power-hereditary because products of
nonnegative matrices are nonnegative.  Progressions of this shape
have a history in the spectral theory of eventual nonnegativity: within
their proof of Friedland's lemma, Zaslavsky and Tam pass to exponents
$m\ell_0r+1$ where $m$ is a product of the orders of all root-of-unity
ratios in the spectrum---exactly a torsion-killing modulus with offset
$1$~\cite[Theorem~3.1]{zaslavsky-tam} (see
also~\cite[Theorem~4.8]{hogben-wilson}).  There the progression aligns
Frobenius multisets; here its job is algorithmic.  We include the
short proof because
applying the principle at the torsion-killing modulus of
Lemma~\ref{lem:torsion} is central to
the algorithm.

\begin{proof}
The forward implication is immediate.  Conversely, suppose the right-hand
side holds for all $k\geq K$ and put $a=DK+1$, $b=D(K+1)+1$.  Then
$A^a,A^b\geq0$ and $\gcd(a,b)=\gcd(DK+1,D)=1$.  Products of nonnegative
matrices are nonnegative, so $A^{pa+qb}=(A^a)^p(A^b)^q\geq0$ for
$p,q\geq0$, and by the Frobenius coin theorem every sufficiently large
integer has this form.
\end{proof}

The lemma is special to matrix powers.
Its scalar analogue would read: if an LRS satisfies $u_{Dk+1}\geq0$ for
every sufficiently large $k$, then $u_n\geq0$ for every sufficiently
large $n$.  This is false: for $D=2$, the rational LRS $u_n=1-2(-1)^n$
equals $3$ at every odd index and $-1$ at every even index, so it is
nonnegative along the entire progression $2k+1$ yet negative infinitely
often.  The set of indices at which an LRS is nonnegative lacks the
additive structure of the set of exponents at which a matrix power is
nonnegative.

\section{Conventions and two general lemmas}\label{sec:prelim}

Throughout, $A\in\Q^{d\times d}$ with $d\geq1$ denotes the input
matrix, $d$ its dimension, and $\langle A\rangle$ the total bit length of
its entries, each rational entry being encoded in reduced form with
numerator and denominator in binary; ``the input size'' always refers to
$\langle A\rangle$.  Since $d\leq\langle A\rangle$, a bound polynomial in
$\langle A\rangle$ and $d$ is simply polynomial in the input size.
We write $e_i$ for the $i$th standard basis vector.  An LRS is
\emph{nondegenerate} if no ratio of two distinct characteristic roots is
a root of unity, and \emph{$\Q_+$-rational} if it admits a representation
$u_n=x^\top M^ny$ with $x,M,y$ entrywise nonnegative and rational.

Algebraic numbers are represented in the usual bit model, by a minimal
polynomial together with an isolating disc selecting the intended root.
We write $h(\alpha)$ for the absolute logarithmic Weil height.  Isolating
discs can be refined, and equality and order of real algebraic numbers
decided, in time polynomial in the representation lengths and the
requested precision~\cite{lll,yap}.  Derived quantities in a root
field $\Q(\lambda)$ are represented in the power basis
$1,\lambda,\ldots,\lambda^{\deg\lambda-1}$ with one common integer
denominator, alongside the primitive representation of $\lambda$ itself.
``Requested precision'' always means a number of correct output bits,
i.e.\ an absolute error bound of the form $2^{-\text{precision}}$.

The first lemma converts an explicit degree-and-height bound into an
exactness test.  The second is the algebraic source of the
degree bounds it consumes.  We state it independently of matrices, since
it applies to any sum of algebraic numbers whose terms are permuted
compatibly with a Galois action, even when no small common field for the
terms exists.

\begin{lemma}[Certified zero and sign tests]\label{lem:certified}
Let $\alpha$ be an algebraic number, and suppose we are given
positive integers $e,H$ with
\[
 [\Q(\alpha):\Q]\leq e,
 \qquad
 h(\alpha)\leq H\log 2,
\]
together with an algorithm that, for every $\varepsilon>0$, computes
a complex ball of radius at most $\varepsilon$ containing $\alpha$.  Then
one can decide whether $\alpha=0$ and, when $\alpha$ is known to be real,
determine its sign, using $O(eH)$ bits of precision and time polynomial
in $e$, $H$, and the cost of one ball evaluation at that precision.
\end{lemma}

\begin{proof}
The Liouville inequality~\cite[Chapter~1]{bombieri-gubler} gives, for
$\alpha\neq0$ and in every embedding,
$|\alpha|\geq\exp(-[\Q(\alpha):\Q]\,h(\alpha))\geq2^{-eH}=:\delta$.
Thus $O(eH)$ bits of precision suffice: approximate $\alpha$ with
absolute error less than $\delta/4$ and
compare the centre's modulus with $\delta/2$.  If $\alpha$ is real and
nonzero, the real part of the centre has its sign.
\end{proof}

In every application below, $e$ and $H$ are polynomially bounded in
the input size and the ball evaluation runs in time polynomial in the
input size and the requested precision, so each invocation of the lemma
takes deterministic polynomial time; we use it in this form without
further comment.

\begin{lemma}[Orbit bound for a block sum]\label{lem:orbit}
Let $K/\Q$ be a Galois extension with group $G$, let $I$ be a finite set
on which $G$ acts, and let $\{\alpha_i\}_{i\in I}\subset K$ be an
\emph{equivariant family}: $\sigma(\alpha_i)=\alpha_{\sigma\cdot i}$ for
all $\sigma\in G$ and all $i\in I$.  For nonempty $J\subseteq I$ define the block
sum $\beta_J=\sum_{i\in J}\alpha_i$.  Then every conjugate of $\beta_J$
over $\Q$ is again a block sum:
\[
 \sigma(\beta_J)=\beta_{\sigma(J)}
 \qquad(\sigma\in G).
\]
Consequently
\[
 [\Q(\beta_J):\Q]
 =|G\cdot\beta_J|
 \leq\bigl|\{\sigma(J):\sigma\in G\}\bigr|,
 \qquad
 h(\beta_J)\leq\sum_{i\in J}h(\alpha_i)+(|J|-1)\log2.
\]
In particular, if $J$ is a block of a $G$-invariant partition of $I$ into
$m$ parts, then $[\Q(\beta_J):\Q]\leq m\leq|I|$.
\end{lemma}

\begin{proof}
Equivariance gives
$\sigma(\beta_J)=\sum_{i\in J}\alpha_{\sigma\cdot i}=\beta_{\sigma(J)}$,
since $i\mapsto\sigma\cdot i$ is a bijection of $I$.  Because $K/\Q$ is
Galois, the conjugates of $\beta_J$ are exactly its $G$-orbit (its minimal
polynomial is $\prod_{\beta\in G\cdot\beta_J}(X-\beta)$), contained
in the block sums $\{\beta_{\sigma(J)}:\sigma\in G\}$.  The orbit--degree
correspondence gives the displayed bound, and a $G$-invariant partition
has only $m$ blocks to which the images $\sigma(J)$ can belong.  The
height inequality is $h(x+y)\leq h(x)+h(y)+\log2$ applied $|J|-1$
times~\cite[Chapter~1]{bombieri-gubler}.
\end{proof}

Note the contrast with the two naive degree bounds.  A sum of
$|J|$ algebraic numbers of degree at most $e$ can require its compositum,
of degree up to $e^{|J|}$, and the only bound available inside a common
Galois field $K$ is $[K:\Q]$, which for the splitting field of a
degree-$d$ polynomial can be as large as $d!$.  Equivariance replaces both
by the number of blocks, because the conjugates of the sum are constrained
to respect the same block structure: the terms are handled locally, in
their own fields, while the symmetry argument certifies globally that
their sum has low degree.  Together with Lemma~\ref{lem:certified} this
gives a reusable testing primitive.

\begin{corollary}[Certified tests for block sums]\label{cor:blocktest}
In the situation of Lemma~\ref{lem:orbit}, suppose that a $G$-invariant
partition of $I$ into at most $m$ blocks is explicitly given, that $m$
and $|I|$ are polynomially bounded in the input size, that each
$\alpha_i$ can be evaluated by certified complex-ball arithmetic in time
polynomial in the input size and the requested precision, and that an
integer $H$, polynomially bounded in the input size, is computable with
$h(\alpha_i)\leq H\log2$ for every $i$.  Then in deterministic polynomial
time one can decide, for every block $J$, whether $\beta_J=0$ and, when
$\beta_J$ is known to be real, its sign.
\end{corollary}

\begin{proof}
Lemma~\ref{lem:orbit} supplies $[\Q(\beta_J):\Q]\leq m$ and
$h(\beta_J)\leq|I|(H+1)\log2$; both bounds are polynomially bounded in
the input size, hence so is their product.
Certified balls for the terms are added, and Lemma~\ref{lem:certified}
applies to the sum.
\end{proof}

\section{Torsion classes and the grouped expansion}\label{sec:expansion}

Let $\Lambda^\times$ be the set of distinct nonzero eigenvalues of $A$.
For $\lambda,\gamma\in\Lambda^\times$, write $\lambda\sim\gamma$ when
$\lambda/\gamma$ is a root of unity; the equivalence classes are called
\emph{torsion classes}.

\begin{lemma}[Torsion classes and modulus]\label{lem:torsion}
In deterministic polynomial time one can compute the torsion classes and
an even positive integer $D$, written in binary with polynomially many
bits, such that (i) $\lambda^D=\gamma^D$ if and only if
$\lambda\sim\gamma$, and (ii) every real number among the values
$\lambda^D$ is strictly positive.
\end{lemma}

\begin{proof}[Proof sketch]
A quotient $\lambda/\gamma$ has algebraic degree at most $d^2$, so if it
is a root of unity of order $r$ then $\varphi(r)\leq d^2$, and the
elementary bound $\varphi(r)\geq\sqrt{r/2}$ gives $r\leq2d^4$.  The order
can therefore be found by testing $\lambda^r=\gamma^r$ for
$1\leq r\leq2d^4$ using Lemma~\ref{lem:certified}; the pairwise decisions
determine the classes.  Take $D=2L$ for the least common multiple $L$ of the orders
found; then $\log D=O(d^2\log d)$.  The factor $2$ turns real $D$th
powers positive: if $\lambda^D$ is real and $\lambda$ is not, then
$\lambda^m$ is real for the order $m$ of $\lambda/\overline\lambda$, which
divides $L$, and $\lambda^D=(\lambda^m)^{2L/m}>0$.  Full details are in
Appendix~\ref{app:torsion}.
\end{proof}

For $\lambda\in\Lambda^\times$, let $E_\lambda$ denote the spectral
projector onto the generalized $\lambda$-eigenspace, and set
$H_{\lambda,s}=(A-\lambda I)^sE_\lambda$ for $0\leq s<d$; this vanishes
once $s$ reaches the largest $\lambda$-Jordan block.  For a torsion class
$C$, define
\begin{equation}\label{eq:B-definition}
 B_{C,s}=\sum_{\lambda\in C}\lambda^{1-s}H_{\lambda,s},
 \qquad
 b_{C,s}^{ij}=(B_{C,s})_{ij},
\end{equation}
and write $\mu_C=\lambda^D$ for any $\lambda\in C$---the \emph{class
root} of $C$---well-defined by Lemma~\ref{lem:torsion}.  The form of the expansion below is what makes the progression tractable:
at $n=Dk+1$ the exponential factor separates as
$\lambda^{Dk+1-s}=(\lambda^D)^k\lambda^{1-s}$, so the huge power appears
only in the common base $\mu_C^k$, which is never constructed, and the
coefficient is the $k$-independent quantity $b_{C,s}^{ij}$
of~\eqref{eq:B-definition}.  The same separation holds at any offset $c$,
with $\lambda^{c-s}$ in place of $\lambda^{1-s}$; the offset $1$ is
chosen because consecutive exponents are then coprime
(Lemma~\ref{lem:coprime-residue}), and the offset $0$ reappears in
Corollary~\ref{cor:existspower}.

We use the following standard uniqueness fact
(proof in Appendix~\ref{app:exppoly}).

\begin{lemma}[Uniqueness of exponential--polynomial forms]\label{lem:exppoly}
Let $\nu_1,\ldots,\nu_m$ be distinct nonzero complex numbers and let
$P_1,\ldots,P_m$ be complex polynomials.  If
$w_k=\sum_{\ell=1}^m P_\ell(k)\,\nu_\ell^k$ vanishes for every
sufficiently large integer $k$, then every $P_\ell$ is identically zero.
Consequently a sequence has at most one such representation with distinct
nonzero bases and nonzero coefficient polynomials, up to reordering.
\end{lemma}

The next lemma defines a closed form for the entry sequence along the progression. The classical spectral expansion writes each entry of $A^n$ as
$\sum_\lambda\sum_s\binom ns\lambda^{n-s}(H_{\lambda,s})_{ij}$: one term
per eigenvalue, with the binomial factors recording the contribution of
nontrivial Jordan blocks.  Substituting $n=Dk+1$ and separating the
exponential factor as above merges all eigenvalues of a
torsion class into a single term, since they share the base
$\mu_C=\lambda^D$.  The result is an exponential--polynomial expression
in $k$ whose bases are the class roots---nondegenerate by the choice of
$D$---and whose coefficients are exactly the grouped coefficients
$b_{C,s}^{ij}$ of~\eqref{eq:B-definition}.  By
Lemma~\ref{lem:exppoly}, this representation is the unique one for the
entry sequence along the progression, so the $b_{C,s}^{ij}$ are intrinsic
to the sequence rather than artifacts of the construction: the sign
criterion of the next section can safely be phrased in terms of them.

\begin{lemma}[Grouped expansion]\label{lem:expansion}
For every $k$ with $Dk+1\geq d$,
\begin{equation}\label{eq:grouped-expansion}
 (A^{Dk+1})_{ij}
 =\sum_C Q_C^{ij}(k)\mu_C^k,
 \qquad
 Q_C^{ij}(X)=\sum_{s=0}^{d-1}\binom{DX+1}{s}b_{C,s}^{ij}.
\end{equation}
The distinct class roots $\mu_C$ form a nondegenerate set.  Moreover,
$Q_C^{ij}=0$ if and only if $b_{C,s}^{ij}=0$ for all $s$; and if
$r=\max\{s:b_{C,s}^{ij}\neq0\}$, then the leading coefficient of
$Q_C^{ij}$ is $D^rb_{C,r}^{ij}/r!$.
\end{lemma}

\begin{proof}
On the generalized $\lambda$-eigenspace the binomial theorem gives
$A^nE_\lambda=\sum_s\binom ns\lambda^{n-s}H_{\lambda,s}$; substitute
$n=Dk+1$ and group those $\lambda$ with equal $D$th powers to
obtain~\eqref{eq:grouped-expansion}.  The generalized $0$-eigenspace is
nilpotent of index at most $d$ and vanishes for $Dk+1\geq d$.  If
$\mu_C/\mu_{C'}$ were a root of unity, representatives $\lambda\in C$ and
$\gamma\in C'$ would satisfy $\lambda\sim\gamma$, contrary to $C\neq C'$.
Finally, the polynomials $\binom{DX+1}{s}$ have distinct degrees and
leading coefficients $D^s/s!$, so the top index $r$ with
$b_{C,r}^{ij}\neq0$ cannot be cancelled by lower-index terms.
\end{proof}

The scalars $b_{C,s}^{ij}$ are called the \emph{grouped coefficients} of
the entry sequence, since each one sums the spectral contributions of a
whole torsion class.  The elements of a torsion class share a common
modulus, since torsion ratios have modulus $1$; we call it the
\emph{modulus} of the class.  Call a class $C$ \emph{$(i,j)$-active} if some
$b_{C,s}^{ij}$ is nonzero, and \emph{$(i,j)$-dominant} if its modulus is
maximal among the active classes.  Equivalently, $\mu_C$ is a dominant
characteristic root of the entry LRS.

\section{A structural criterion}\label{sec:criterion}

We use the following classical consequence of Berstel's theorem for
rational series over the nonnegative rationals.

\begin{lemma}[Nonnegative rational sequences]\label{lem:berstel}
Let $(u_n)$ be a nondegenerate rational LRS, not eventually zero, which
admits a linear representation $u_n=x^\top M^ny$ with $x,M,y$ entrywise
nonnegative and rational.  Then $(u_n)$ has a unique characteristic root
of maximum modulus.
\end{lemma}

This is Proposition~6 of~\cite{dcosta-ow}, deriving from Berstel's
theorem~\cite{berstel-1971,berstel-reutenauer}: for a series with a
nonnegative realization, the poles of least modulus occur in a finite
periodic pattern, so the dominant roots of the recurrence differ by roots
of unity, and nondegeneracy leaves only one.  Excluding eventually-zero
sequences is essential: their generating series are polynomials, the
exceptional pole-free case of Berstel's statement; we detect that branch
separately as case (1) below.  Although Berstel's theorem is stated
for $\N$-rational series, the nonnegative-rational case reduces to it by
scaling: if $N$ is a common denominator for the entries of $x$, $M$, and
$y$, then $N^{n+2}u_n=(Nx)^\top(NM)^n(Ny)$ is $\N$-rational, and the
scaling multiplies every characteristic root by $N$, affecting neither
nondegeneracy nor the uniqueness of a maximum-modulus root.

\begin{theorem}[Exact criterion]\label{thm:criterion}
The matrix $A$ is eventually nonnegative if and only if, for every pair
$(i,j)$, one of the following holds:
\begin{enumerate}
\item $b_{C,s}^{ij}=0$ for every $C,s$; or
\item there is exactly one $(i,j)$-dominant class $C_*$, and, writing
      $r=\max\{s:b_{C_*,s}^{ij}\neq0\}$, one has $b_{C_*,r}^{ij}>0$.
\end{enumerate}
In the second case $C_*$ is stable under complex conjugation,
$b_{C_*,r}^{ij}$ is real, and $\mu_{C_*}>0$.
\end{theorem}

\begin{proof}[Proof sketch]
Fix an entry $u_k=(A^{Dk+1})_{ij}$.  Conjugating
\eqref{eq:grouped-expansion} and applying Lemma~\ref{lem:exppoly}
termwise gives $\overline{b_{C,s}^{ij}}=b_{\overline C,s}^{ij}$; hence a
unique maximum-modulus active class $C_*$ is conjugation-stable, so
$\mu_{C_*}>0$ by Lemma~\ref{lem:torsion} and every $b_{C_*,s}^{ij}$ is
real.

Sufficiency is the dominant-term argument: with one active class of
strictly largest modulus, $u_k/(k^r\mu_{C_*}^k)\to
D^rb_{C_*,r}^{ij}/r!>0$, so $u_k>0$ eventually; taking the maximum
threshold over entries and applying Lemma~\ref{lem:coprime-residue} gives
eventual nonnegativity.

For necessity, eventual nonnegativity must exclude ties between distinct
classes of maximal modulus, which are a priori possible.  Choose $t$ so
large that $A^{Dt}$ and $A^{Dt+1}$ are both nonnegative, and sample the
entry sequence at $k=t(m+1)$: the sequence
$v_m=e_i^\top A^{Dt+1}(A^{Dt})^me_j$ has a nonnegative rational
realization, so Lemma~\ref{lem:berstel} forces a unique dominant root,
and passing to the subsequence merges no nondegenerate roots.  Hence $u$
itself has one dominant class, and its eventual nonnegativity forces the
leading grouped coefficient to be positive.  The full proof is in
Appendix~\ref{app:criterion}.
\end{proof}

The criterion is essentially the effective characterization
of~\cite{dcosta-ow}: the LRS-theoretic core of the necessity argument
is Propositions~4--6 there.  What is new is its phrasing on the single
progression $Dk+1$ in terms of the grouped coefficients
$b_{C,s}^{ij}$---exactly what the algorithm tests---and that this form
can be checked in polynomial time, via the reduction to one residue
class and Galois equivariance in place of the splitting field.

A small instance: the symmetric matrix $A_\star$ of
Appendix~\ref{app:example}, with eigenvalues $11/10$, $1$, $-1$, has
torsion classes $\{11/10\}$ and $\{1,-1\}$, and $D=4$.  For every entry
the class $\{11/10\}$ is the unique dominant one and its leading grouped
coefficient is $11/30>0$, so the criterion accepts---even though the
degenerate pair $1,-1$ makes individual powers of $A_\star$ carry
negative entries up to exponent $7$.

\section{Computing the grouped coefficients}\label{sec:grouped}

The apparent obstruction is that $D$ can be exponentially large as a
value.  In general $h(\lambda^D)=D\,h(\lambda)$, so constructing $A^D$,
any $\lambda^D$, or the splitting field of $\chi_A$ (of degree up to $d!$)
can require exponentially many bits.  We now show that none of these objects is needed.
The computation follows the pipeline
\[
 \lambda
 \;\longrightarrow\;
 \text{torsion class }C
 \;\longrightarrow\;
 T_{\lambda,s}^{ij}\in\Q(\lambda)
 \;\longrightarrow\;
 b_{C,s}^{ij}=\textstyle\sum_{\lambda\in C}T_{\lambda,s}^{ij}
 \;\longrightarrow\;
 \text{dominant class},
\]
with each arrow computable in polynomial time; here
$T_{\lambda,s}^{ij}=\lambda^{1-s}\bigl((A-\lambda I)^sE_\lambda\bigr)_{ij}$
denotes the contribution of the single eigenvalue $\lambda$ to the
grouped coefficient $b_{C,s}^{ij}$, introduced formally
in~\eqref{eq:T-definition} below.

\subsection{Spectral projectors over one root field}
Let $\chi_A(X)$ be the characteristic polynomial.  If $\lambda$ has
algebraic multiplicity $m_\lambda$, set
$q_\lambda(X)=\chi_A(X)/(X-\lambda)^{m_\lambda}\in\Q(\lambda)[X]$, and let
$h_\lambda$ be the inverse of $q_\lambda$ modulo
$(X-\lambda)^{m_\lambda}$, computed by its truncated Taylor expansion at
$\lambda$.  Put $p_\lambda=q_\lambda h_\lambda$.

\begin{lemma}[Explicit projector]\label{lem:projector}
The spectral projector is $E_\lambda=p_\lambda(A)$.  All coefficients of
$p_\lambda$ and all entries of $E_\lambda$ lie in $\Q(\lambda)$ and have
exact representations of polynomial size and polynomial logarithmic height
in the input size.
\end{lemma}

\begin{proof}[Proof sketch]
By construction $p_\lambda\equiv1\pmod{(X-\lambda)^{m_\lambda}}$ and
$p_\lambda\equiv0\pmod{(X-\gamma)^{m_\gamma}}$ for every other eigenvalue
$\gamma$, so $p_\lambda(A)$ is the Chinese-remainder idempotent acting as
the identity on every $\lambda$-Jordan block and as zero elsewhere.  The
Taylor inverse is computed only in the degree-at-most-$d$ field
$\Q(\lambda)$: expanded in a power basis, its defining equations form one
rational linear system of dimension at most $d^2$ with polynomial-bit
coefficients, so Gaussian elimination gives polynomial-bit coordinates.
This is the basis-free form of the root-field-by-root-field computation in
Cai's exact Jordan algorithm~\cite{cai}.  Details are in
Appendix~\ref{app:projector}.
\end{proof}

It follows that every summand
\begin{equation}\label{eq:T-definition}
 T_{\lambda,s}^{ij}
 =\bigl(\lambda^{1-s}(A-\lambda I)^sE_\lambda\bigr)_{ij}
 \in\Q(\lambda)
\end{equation}
has a polynomial-size exact representation.  The certified tests
need one further ingredient.  To set its working precision,
Lemma~\ref{lem:certified} must be \emph{given} an integer $H$ bounding
the height of the number under test.  We therefore compute, for each
summand, an explicit integer $\widehat H_{\lambda,s}^{ij}$ with
\[
 h(T_{\lambda,s}^{ij})\leq\widehat H_{\lambda,s}^{ij}\log2 .
\]
This height certificate is assembled from two standard estimates.
Mignotte's factor bound controls the minimal polynomials of the
eigenvalues.  Height inequalities for sums and products then control
the arithmetic that builds $T_{\lambda,s}^{ij}$ from
them~\cite{bombieri-gubler,yap}.  The resulting
$\widehat H_{\lambda,s}^{ij}$ is polynomially bounded in the input
size; the calculation is in Appendix~\ref{app:heights}.

\subsection{Low degree of the class sum}
By~\eqref{eq:B-definition} and~\eqref{eq:T-definition}, each grouped
coefficient is a \emph{class sum} of the local contributions:
\begin{equation}\label{eq:classsum}
 b_{C,s}^{ij}=\sum_{\lambda\in C}T_{\lambda,s}^{ij},
\end{equation}
a sum over one torsion class of terms that may lie in different root
fields.  

\begin{lemma}[Grouped coefficient lemma]\label{lem:grouped}
Let $K$ be a splitting field of $\chi_A$ and $G=\Gal(K/\Q)$.  For every
$\sigma\in G$, the class sums satisfy the equivariance law
$\sigma(B_{C,s})=B_{\sigma(C),s}$.  For every torsion class $C$, every
$s<d$, and every $(i,j)$,
\[
 [\Q(b_{C,s}^{ij}):\Q]\leq d,
 \qquad
 h(b_{C,s}^{ij})\leq\poly(\langle A\rangle,d).
\]
If $C$ is stable under complex conjugation, then $b_{C,s}^{ij}$ is real.
\end{lemma}

\begin{proof}
Automorphisms permute the eigenvalues and preserve the torsion
relation, so $G$ permutes the classes together with their class roots:
$\sigma(\mu_C)=\sigma(\lambda)^D=\mu_{\sigma(C)}$.  The equivariance
law now follows by conjugating the grouped expansion, exactly as in the
complex-conjugation step of the proof of Theorem~\ref{thm:criterion}.
Fix $(i,j)$ and apply $\sigma$ to the
identity~\eqref{eq:grouped-expansion}: the left-hand side
$(A^{Dk+1})_{ij}$ is rational, hence fixed, while the right-hand side
becomes $\sum_C(\sigma Q_C^{ij})(k)\,\mu_{\sigma(C)}^k$, since $\sigma$
fixes the integer $k$ and the rational coefficients of the binomial
polynomials.  This is a second exponential--polynomial representation
of the same sequence, with distinct nonzero bases, so
Lemma~\ref{lem:exppoly} forces
$\sigma(Q_C^{ij})=Q_{\sigma(C)}^{ij}$; comparing coefficients in the
basis $\binom{DX+1}{s}$ gives
$\sigma(b_{C,s}^{ij})=b_{\sigma(C),s}^{ij}$, and hence entrywise
$\sigma(B_{C,s})=B_{\sigma(C),s}$.  

For the degree bound, apply Lemma~\ref{lem:orbit} to the family
$\{b_{C,s}^{ij}\}_C$, indexed by the at most $d$ torsion classes with
the $G$-action just described and partitioned into singletons: every
conjugate of $b_{C,s}^{ij}$ is again one of the at most $d$ class
sums, so $[\Q(b_{C,s}^{ij}):\Q]\leq d$.  The height bound follows
from~\eqref{eq:classsum}, the bounds of~\eqref{eq:T-definition}, and
$h(x+y)\leq h(x)+h(y)+\log2$, as carried out in
Appendix~\ref{app:heights}.  Complex conjugation is one of the
automorphisms, so a conjugation-stable class has real coefficients.
\end{proof}

\begin{lemma}[Zero and sign tests for grouped coefficients]\label{lem:separation}
There is a deterministic polynomial-time algorithm which, given
$C,s,i,j$, decides whether $b_{C,s}^{ij}=0$ and, whenever $b_{C,s}^{ij}$
is real, its sign.
\end{lemma}

\begin{proof}
This is Corollary~\ref{cor:blocktest} with the torsion classes as
the index set: by Lemma~\ref{lem:grouped}, the family
$\{b_{C,s}^{ij}\}_C$, indexed by the at most $d$ classes, is
equivariant, and we take the (trivially $G$-invariant) partition of
this index set into singletons, so that the block sums are the grouped
coefficients themselves.  The height hypothesis is supplied by the
integers $\widehat H_{C,s}^{ij}$ of Appendix~\ref{app:heights}.  For
the evaluation hypothesis, $b_{C,s}^{ij}$ is evaluated through its
class, stored as a list of isolated roots.  Each contribution
$T_{\lambda,s}^{ij}$ is first computed \emph{exactly}, as
$T_{\lambda,s}^{ij}=g(\lambda)/q$ with $g\in\Z[X]$ and $q\in\Z_{>0}$
(Lemma~\ref{lem:projector} and Appendix~\ref{app:heights}); all
divisions, including those by $\lambda$ and by $q_\lambda(\lambda)$ in
the projector formula, happen at this exact symbolic stage.  The
numerical stage is division-free: evaluate the integer polynomial $g$
on a refined isolating disc for $\lambda$, divide the resulting ball
by the integer $q$, and sum over the class.  Since $g$ and $q$ have
polynomial bit length, $p$ correct bits of $b_{C,s}^{ij}$ cost time
polynomial in the input size and $p$, and
Corollary~\ref{cor:blocktest} concludes.
\end{proof}

An illustration of the phenomenon being tested: the
companion matrix
$R=\begin{psmallmatrix}0&0&2\\1&0&0\\0&1&0\end{psmallmatrix}$ of $X^3-2$
has eigenvalues $\sqrt[3]2$, $\omega\sqrt[3]2$,
$\overline\omega\sqrt[3]2$, forming a single torsion class but lying in
three distinct cubic fields.  In the $(1,1)$ entry the three local
contributions $\lambda/3$, each evaluated in its own field $\Q(\lambda)$,
cancel exactly, while in the $(1,3)$ entry they sum to $2$.  The certified test detects the cancellation by summing balls computed
independently in the three embeddings, never constructing the degree-$6$
splitting field.  We work through the example in more detail in Appendix~\ref{app:example}.

\section{The algorithm}\label{sec:algorithm}

The input is a matrix $A\in\Q^{d\times d}$, each entry given by its
binary encoding; no further assumptions are made---$A$ may be singular,
reducible, or non-diagonalizable.  The output is the decision whether
$A$ is eventually nonnegative.  Throughout, $D$ and the torsion classes
$C$ are those of Lemma~\ref{lem:torsion}, and the grouped coefficients
$b_{C,s}^{ij}$ are those of~\eqref{eq:B-definition}; the algorithm
tests the criterion of Theorem~\ref{thm:criterion} entry by entry.

\begin{enumerate}
\item Compute the distinct nonzero eigenvalues, their multiplicities, and
      the torsion classes (Lemma~\ref{lem:torsion}), stored as isolating
      discs grouped by class.
\item For every $C,s,i,j$, decide whether $b_{C,s}^{ij}=0$
      (Lemma~\ref{lem:separation}).
\item Compare the moduli of active classes by testing
      $\lambda\overline\lambda-\gamma\overline\gamma$ for class
      representatives: the difference is real, has degree at most $d^4$
      and polynomial height, so Lemma~\ref{lem:certified} decides its zero
      case and its sign.  No conjugate $\overline\lambda$ is isolated as
      an algebraic number: refining the isolating disc of $\lambda$
      certifies $|\lambda|$ directly, and the degree and height bounds
      enter only through the precision required by
      Lemma~\ref{lem:certified}.
\item For each entry $(i,j)$: accept the entry if no class is active;
      otherwise reject it unless there is a unique active class $C_*$ of
      maximum modulus; take the largest $s$ with $b_{C_*,s}^{ij}\neq0$ and
      accept the entry exactly when $b_{C_*,s}^{ij}>0$.  No conjugation
      check is needed: uniqueness forces $C_*=\overline{C_*}$
      (Theorem~\ref{thm:criterion}), so the coefficient is real and
      Lemma~\ref{lem:separation} decides its sign.
\item Accept $A$ exactly when every entry is accepted.
\end{enumerate}

Correctness is Theorem~\ref{thm:criterion}.  There are polynomially many
roots, classes, indices, and Jordan degrees, and every subcomputation is
polynomial by Lemmas~\ref{lem:torsion}, \ref{lem:projector},
\ref{lem:grouped}, and~\ref{lem:separation}.  No $D$th power of a matrix
or algebraic number, and no global splitting field, is constructed.  Hence
the running time is polynomial in the bit length of $A$, which proves
Theorem~\ref{thm:main}.

\begin{remark}[Eventual positivity is also in P]\label{rem:positivity}
Strict positivity is likewise closed under products, so the same
algorithm, with case (1) of Theorem~\ref{thm:criterion} rejected, decides
eventual \emph{positivity} in deterministic polynomial time: a
self-contained test, independent of the Perron--Frobenius characterization
of~\cite{noutsos}.
\end{remark}

\begin{remark}[A valid threshold has polynomially many bits]\label{rem:threshold}
The algorithm does not compute the least eventuality threshold, which
depends on the exact finite transient.  On a yes-instance it can
nevertheless certify \emph{a} threshold $N$ with polynomially many
bits, i.e.\ $A^n\geq0$ for every $n\geq N$.  The degree and height bounds of Section~\ref{sec:grouped} yield an explicit crossover bound
$K=2^{\poly}$ with $A^{Dk+1}\geq0$ for all $k\geq K$, and the Frobenius
bound applied to the coprime pair $DK+1$, $D(K+1)+1$ then converts $K$
into $N$.  The calculation is in Appendix~\ref{app:threshold}.
\end{remark}

The last result of this section concerns the mere \emph{existence}
of a nonnegative power.  Matrices with at least one nonnegative power
are called \emph{power nonnegative}; their spectral structure, and its
relation to eventual nonnegativity, is studied by Tudisco, Cardinali,
and Di Fiore~\cite{tudisco-cdf}.  The decision procedure rests on a
structural fact worth isolating:
a single nonnegative power, however large its exponent, already forces
nonnegativity along the whole progression $Dk$---and the modulus $D$ is
determined by the spectrum of $A$ alone (Lemma~\ref{lem:torsion}), not
by the witness.

\begin{lemma}[One nonnegative power forces the whole progression]\label{lem:onepower}
If $A^{n_0}\geq0$ for some $n_0\geq1$, then $A^{Dk}\geq0$ for all
sufficiently large $k$.
\end{lemma}

\begin{proof}[Proof sketch]
$A^{Dn_0}=(A^{n_0})^D\geq0$, so every entry's subsequence along
$\{Dn_0m\}$ is $\Q_+$-rational.  For each entry that is not eventually
zero, Berstel's theorem applied to this subsequence yields the unique
dominant class with positive leading coefficient that the criterion of
Theorem~\ref{thm:criterion} requires along the progression $Dk$, and
the sufficiency direction of the criterion then gives $A^{Dk}\geq0$
for all large $k$.  The full proof is in
Appendix~\ref{app:existspower}.
\end{proof}

\begin{proposition}[Offset-$0$ test]\label{prop:offsetzero}
Whether $A^{Dk}\geq0$ for all sufficiently large $k$ is decidable in
deterministic polynomial time.
\end{proposition}

\begin{proof}[Proof sketch]
The grouped expansion holds verbatim along $Dk$, with the same
nondegenerate class roots $\mu_C$ and with the offset-$0$ coefficients
$b_{C,s}^{ij}(0)=\sum_{\lambda\in C}\lambda^{-s}(H_{\lambda,s})_{ij}$.
The criterion of Theorem~\ref{thm:criterion} transfers to offset $0$: sufficiency is the same asymptotic argument, and necessity
follows from the proof of Lemma~\ref{lem:onepower}.  The coefficients
obey the equivariance law of Lemma~\ref{lem:grouped}---conjugate the
offset-$0$ expansion---so Lemma~\ref{lem:separation} decides their
zero and sign, and the algorithm of Section~\ref{sec:algorithm} runs
unchanged with these coefficients
(Appendix~\ref{app:existspower}).
\end{proof}

\begin{corollary}[Power nonnegativity is in P]\label{cor:existspower}
One can decide in deterministic polynomial time whether a given
rational matrix is power nonnegative, i.e., whether there exists
$n\geq1$ with $A^n\geq0$.
\end{corollary}

\begin{proof}
If $A^{Dk}\geq0$ for some $k$ with $Dk\geq1$, then some power is
nonnegative; conversely, Lemma~\ref{lem:onepower} converts any
nonnegative power into eventual nonnegativity along $Dk$.  So the
answer is yes exactly when the test of
Proposition~\ref{prop:offsetzero} accepts.
\end{proof}

Corollary~\ref{cor:existspower} upgrades the effective
characterization of~\cite[Remark~8]{dcosta-ow} to a polynomial-time
procedure.  It settles the one-matrix case of the nonnegativity
problems studied there for semigroups generated by several commuting
matrices.

\section{Conclusion}\label{sec:conclusion}

Eventual nonnegativity of a rational matrix is decidable in polynomial
time: a single coprime-offset progression determines the structure of the power sequence,
and the spectral criterion along it is checked by exact computation in
individual root fields, not in the splitting field of the characteristic polynomial.  We note three
directions for further work.

First, \emph{commuting families}.  For semigroups generated by several
commuting matrices, the existence of a nonnegative matrix in the semigroup is decidable
subject to Schanuel's conjecture~\cite{dcosta-ow}. Can one decide
whether, for commuting generators $A_1,\ldots,A_m$, there exists $N$
such that
$A_1^{n_1}\cdots A_m^{n_m}\geq0$ for all $n_1,\ldots,n_m\geq N$?

Second, the \emph{entry dichotomy}.  Eventual nonnegativity of all
$d^2$ entries is in P by Theorem~\ref{thm:main}; eventual
nonnegativity of one prescribed entry is Ultimate Positivity, whose
decidability is open.  Nothing is known between these extremes, for
example for a prescribed subset of entries.

Third, the \emph{least threshold}.  Our algorithm certifies \emph{a}
threshold with polynomially many bits (Remark~\ref{rem:threshold}), but
the complexity of computing the least eventuality threshold---or of
deciding ``$A^n\geq0$ for all $n\geq N$'' with $N$ given in binary---is
unresolved.

\appendix

\section{Deferred proofs}\label{app:proofs}

\subsection{Lemma~\ref{lem:torsion} (Torsion classes and modulus)}\label{app:torsion}

\begin{proof}
Canonical representations of the distinct roots of the characteristic
polynomial can be computed in polynomial time: factor $\chi_A$ over
$\Q$~\cite{lll} and isolate the roots of each irreducible factor by a
polynomial-time root-isolation procedure~\cite{yap}.  For each
pair $\lambda,\gamma$, the quotient $\eta=\lambda/\gamma$ has degree at
most $d^2$.  If it is a root of unity of order $r$, then
$\varphi(r)\leq d^2$ and $\varphi(r)\geq\sqrt{r/2}$, hence
$r\leq2d^4$.  (The totient inequality follows from
$r/\varphi(r)^2=\prod_{p^a\parallel r}p^{2-a}/(p-1)^2\leq2$: the factor at
$p=2$ is at most $2$, and every odd-prime factor is less than $1$.)  The
order can therefore be found by testing $\lambda^r-\gamma^r=0$ for
$1\leq r\leq2d^4$.  (One may first test whether
$|\lambda|=|\gamma|$---a single certified sign test on
$\lambda\overline\lambda-\gamma\overline\gamma$, as in Step~3 of the
algorithm---and run the loop only on pairs of equal modulus, since a
root-of-unity ratio has modulus $1$.)  The difference lies in $\Q(\lambda,\gamma)$, hence has
degree at most $d^2$, and
\[
 h(\lambda^r-\gamma^r)
 \leq r\bigl(h(\lambda)+h(\gamma)\bigr)+\log2,
\]
which is polynomially bounded in the input size, since $r\leq2d^4$ and the
eigenvalues have polynomial height as roots of $\chi_A$.
Lemma~\ref{lem:certified} therefore decides each equality in polynomial
time.  The least successful $r$ is the exact order, and the pairwise
decisions determine the classes.  Only pairwise calculations are
involved: we never build one field containing \emph{all} eigenvalues,
whose degree could be exponential.

Let $L$ be the least common multiple of the orders found, with
$\operatorname{lcm}(\varnothing)=1$, and put $D=2L$.  There are at most
$d^2$ recorded orders, each at most $2d^4$, so $\log D=O(d^2\log d)$.  If
$\lambda\sim\gamma$, then the order of $\lambda/\gamma$ divides $L$, hence
divides $D$, and $\lambda^D=\gamma^D$.  Conversely, $\lambda^D=\gamma^D$
gives $(\lambda/\gamma)^D=1$, so $\lambda/\gamma$ is torsion.  This proves
the first claim.

For the last claim, suppose $\mu=\lambda^D$ is real.  If $\lambda$ is
real, then $\mu>0$ because $D$ is even.  Otherwise
$\zeta=\lambda/\overline\lambda$ satisfies $\zeta^D=1$.  Since $\chi_A$
has real coefficients, $\overline\lambda$ is also an eigenvalue, so the
order $m$ of $\zeta$ is among the orders recorded in $L$; in particular
$m\mid L$.  Then $\lambda^m=\overline{\lambda^m}$ is real, and
$D/m=2L/m$ is even.  Hence $\lambda^D=(\lambda^m)^{D/m}>0$.  The
underlying canonical-arithmetic and root-of-unity procedures are stated
explicitly in~\cite{ospw,ow-ultimate}.
\end{proof}

\subsection{Lemma~\ref{lem:exppoly} (Uniqueness of exponential--polynomial forms)}\label{app:exppoly}

\begin{proof}
Repeatedly apply the difference operator
$w_k\mapsto w_{k+1}-\nu_m w_k$: each application lowers $\deg P_m$ by one
and, for $\ell<m$, replaces $P_\ell$ by
$\nu_\ell P_\ell(k+1)-\nu_mP_\ell(k)$, which has the \emph{same} degree as
$P_\ell$ because $\nu_\ell\neq\nu_m$.  After $\deg P_m+1$ applications
the $\nu_m$-term is gone and no other degree has changed, so induction on
$m$ finishes.  Equivalently, the relevant confluent Vandermonde
determinant is nonzero.
\end{proof}

\subsection{Theorem~\ref{thm:criterion} (Exact criterion)}\label{app:criterion}

\begin{proof}
We first justify the automatic stability claimed in case (2), using only
that $u_k=(A^{Dk+1})_{ij}$ is a real sequence.
Conjugating~\eqref{eq:grouped-expansion} gives
\[
 u_k=\overline{u_k}
 =\sum_C\overline{Q_C^{ij}}(k)\,\overline{\mu_C}^{\,k}
 =\sum_C\overline{Q_C^{ij}}(k)\,\mu_{\overline C}^{\,k},
\]
where $\overline{\mu_C}=\overline{\lambda^D}=\overline\lambda^{\,D}
=\mu_{\overline C}$.  Lemma~\ref{lem:exppoly} makes this the same
expansion termwise: conjugation permutes the active classes, preserves
their moduli, and sends $Q_C^{ij}$ to
$Q_{\overline C}^{ij}=\overline{Q_C^{ij}}$.  Since the polynomials
$\binom{DX+1}{s}$ have rational coefficients and distinct degrees, this
says $\overline{b_{C,s}^{ij}}=b_{\overline C,s}^{ij}$ for every $C,s,i,j$.
A unique active class $C_*$ of maximum modulus is therefore fixed by
conjugation; then $\mu_{C_*}=\overline{\mu_{C_*}}$ is real, hence $>0$ by
Lemma~\ref{lem:torsion}, and every $b_{C_*,s}^{ij}$ is real.

Suppose first that $A$ is eventually nonnegative.  Choose $t$ sufficiently
large that $Dt\geq d$ and $M_0=A^{Dt}\geq0$, $M_1=A^{Dt+1}\geq0$.  Fix
$(i,j)$ and put $u_k=(A^{Dk+1})_{ij}$ and $v_m=u_{t(m+1)}$.  Then
\[
 v_m=u_{t(m+1)}=\bigl(A^{Dt(m+1)+1}\bigr)_{ij}
 =e_i^\top A^{Dt+1}(A^{Dt})^me_j=e_i^\top M_1M_0^me_j,
\]
so $v$ has a nonnegative rational linear representation.  The nonzero
characteristic roots of $v$ are $\mu_C^t$ for the active classes $C$.
They remain nondegenerate---if $(\mu_C/\mu_{C'})^t$ were a root of unity,
then $\mu_C/\mu_{C'}$ would be one---and a nonzero polynomial $Q_C^{ij}$
remains nonzero after the affine substitution $k=t(m+1)$.  If no class is
active,
$u$ is eventually zero and case (1) holds.  Otherwise
Lemma~\ref{lem:exppoly} shows that $v$ is not eventually zero.  It is a
nondegenerate $\Q_+$-rational sequence, so Lemma~\ref{lem:berstel} gives
it a unique dominant root.  Therefore $u$ has a unique active class $C_*$
of maximum modulus.  By the opening paragraph, $\mu_{C_*}>0$ and
$b_{C_*,r}^{ij}$ is real, and by Lemma~\ref{lem:expansion},
$u_k/(k^r\mu_{C_*}^k)\to D^rb_{C_*,r}^{ij}/r!$.  The left-hand side is
eventually nonnegative and the limit is nonzero, so $b_{C_*,r}^{ij}>0$.

Conversely, suppose the stated condition holds.  In the first case $u_k$
is eventually zero.  In the second case, let
$\rho=\max\{|\mu_C|:C\neq C_*\text{ is active}\}$, with $\rho=0$ if there
is no other active class.  Then
\[
 u_k=\mu_{C_*}^k
 \left(\frac{D^r}{r!}b_{C_*,r}^{ij}k^r+O(k^{r-1})\right)
 +O(k^{d-1}\rho^k).
\]
Since $\rho<\mu_{C_*}$ and $b_{C_*,r}^{ij}>0$, we have $u_k>0$ for all
sufficiently large $k$.  Taking the maximum of the finitely many entrywise
thresholds gives $A^{Dk+1}\geq0$ eventually, and
Lemma~\ref{lem:coprime-residue} gives eventual nonnegativity of $A$.
\end{proof}

\subsection{Lemma~\ref{lem:projector} (Explicit projector)}\label{app:projector}

\begin{proof}
Expand $q_\lambda(\lambda+T)=\sum_{r\geq0}a_rT^r$ with
$a_0=q_\lambda(\lambda)\neq0$, and define recursively $c_0=a_0^{-1}$ and
$c_r=-a_0^{-1}\sum_{t=1}^r a_tc_{r-t}$ for $1\leq r<m_\lambda$; then
$h_\lambda(X)=\sum_{r<m_\lambda}c_r(X-\lambda)^r$ is the inverse Taylor
series of $q_\lambda$ modulo $(X-\lambda)^{m_\lambda}$, so
$p_\lambda=q_\lambda h_\lambda$ satisfies
$p_\lambda\equiv1\pmod{(X-\lambda)^{m_\lambda}}$.  For every other
eigenvalue $\gamma$, the polynomial $q_\lambda$ is divisible by
$(X-\gamma)^{m_\gamma}$, so
$p_\lambda\equiv0\pmod{(X-\gamma)^{m_\gamma}}$.  Evaluation at the Jordan
form therefore shows that $p_\lambda(A)$ is the identity on every
$\lambda$-Jordan block and zero on every other, i.e., $E_\lambda$.

For complexity, arithmetic takes place only in the degree-at-most-$d$
field $\Q(\lambda)$.  To avoid concealing coefficient growth in the
recursion for the $c_r$, expand all equations
$\sum_{t=0}^r a_tc_{r-t}=\delta_{r0}$ ($0\leq r<m_\lambda$)
simultaneously in a power basis for $\Q(\lambda)$.  This is a rational
linear system of dimension at most $d^2$, with polynomial-bit
coefficients.  Gaussian elimination and determinant bounds give
polynomial-bit coordinates for every $c_r$.  Subsequent polynomial and
matrix evaluation preserves a polynomial bit bound.  Equivalently,
standard height inequalities applied to the displayed formulas give
polynomial logarithmic height.  No splitting field is constructed.
\end{proof}

\subsection{Explicit height bounds}\label{app:heights}
Clear denominators in the power-basis output and write
$T_{\lambda,s}^{ij}=g(\lambda)/q$ with $g\in\Z[X]$,
$\deg g<[\Q(\lambda):\Q]$, and $q\in\Z_{>0}$.  If $\widehat h_\lambda$ is
any explicitly computed upper bound for $h(\lambda)$, then
\[
 h(T_{\lambda,s}^{ij})
 \leq\log\max(1,\lVert g\rVert_1)+(\deg g)\,\widehat h_\lambda+\log q.
\]
The irreducible factors of $\chi_A$ have polynomial bit length by
Mignotte's factor bound~\cite{yap}, and if a primitive integral minimal
polynomial has degree $e$ and coefficient height $H$, the standard
Mahler-measure estimate gives the safe bound
$\widehat h_\lambda=\log H+\log(e+1)$.  Every quantity on the right-hand
side has polynomial bit length and polynomial magnitude, so the algorithm
can compute an integer $\widehat H_{\lambda,s}^{ij}$ with
$h(T_{\lambda,s}^{ij})\leq\widehat H_{\lambda,s}^{ij}\log2$; then
\[
 \widehat H_{C,s}^{ij}
 =\sum_{\lambda\in C}\widehat H_{\lambda,s}^{ij}+|C|-1
\]
satisfies $h(b_{C,s}^{ij})\leq\widehat H_{C,s}^{ij}\log2$ and is
polynomially bounded.  The height inequalities used here, and the relation
between a power-basis representation and absolute height, are standard;
see~\cite[Chapter~1]{bombieri-gubler}.

\subsection{Remark~\ref{rem:threshold} (A valid threshold has polynomially many bits)}\label{app:threshold}

\begin{proof}
Fix an entry $(i,j)$ in case (2) of Theorem~\ref{thm:criterion}
and write $Q_C$ for $Q_C^{ij}$, expanded in powers of $X$ rather than
in the binomial basis $\binom{DX+1}{s}$; the conversion factors are
bounded by $(2D)^dd!$ and have polynomial logarithmic size.  The
explicit bounds of Section~\ref{sec:grouped}, together with
Lemma~\ref{lem:certified} applied to the nonzero quantities involved,
therefore give a polynomial $p$ in the input size such that, after
enlarging $p$ if necessary so that $2^p\geq2p+d+4$:
\begin{enumerate}
\item the leading (degree-$r$) coefficient of $Q_{C_*}$ is positive
      and at least $2^{-p}$;
\item the absolute values of the monomial coefficients of all the
      $Q_C$, summed over all classes and degrees, are at most $2^{p}$;
\item $\log_2(\mu_{C_*}/\rho)\geq2^{-p}$ whenever a subdominant active
      class exists ($\rho>0$).
\end{enumerate}
For (3), the separation is inherited from class representatives:
$\log_2(\mu_{C_*}/\rho)=D\log_2(|\lambda_*|/|\gamma|)$ with
$|\lambda_*|>|\gamma|$, the difference $|\lambda_*|^2-|\gamma|^2$ is a
nonzero real algebraic number of polynomial degree and height, hence at
least $2^{-\poly}$, and the exponentially large factor $D\geq1$ only
\emph{widens} the gap.
Write $u_k=(A^{Dk+1})_{ij}=Q_{C_*}(k)\,\mu_{C_*}^k+E_k$ with
$E_k=\sum_{C\neq C_*}Q_C(k)\,\mu_C^k$.

\emph{Main term.}  By (1) and (2), for $k\geq1$,
\[
 Q_{C_*}(k)\;\geq\;2^{-p}k^r-2^{p}k^{r-1}
 \;=\;2^{-p}k^r\Bigl(1-\frac{2^{2p}}{k}\Bigr),
\]
so $k>2^{2p+2}$ gives $Q_{C_*}(k)\geq\tfrac12\,2^{-p}k^r>0$.

\emph{Error term.}  By (2), $|E_k|\leq2^{p}k^{d-1}\rho^k$, and by
(3), $(\rho/\mu_{C_*})^k\leq2^{-k2^{-p}}$.  Hence
$|E_k|<\tfrac12\,2^{-p}k^r\mu_{C_*}^k$, and thus $u_k>0$, as soon as in
addition
\[
 k\,2^{-p}\;>\;2p+2+d\log_2k .
\]

\emph{Crossover.}  Both conditions hold at $k_0=(d+1)2^{2p+2}$: the
first trivially, and the second because
$k_0\,2^{-p}=(d+1)2^{p+2}$, while
\[
 2p+2+d\log_2k_0
 \;\leq\;(d+1)\bigl(2p+2+\log_2(d+1)\bigr)
 \;<\;(d+1)2^{p+2}
\]
by $2^p\geq2p+d+4$ and $\log_2(d+1)\leq d$.  They persist for every
$k\geq k_0$, since the left side of the second condition grows linearly
in $k$ and the right side only logarithmically.  Hence $u_k>0$ for all
$k\geq(d+1)2^{2p+2}$.  When $C_*$ is the only active class, $\rho=0$,
the error term is absent and the main-term condition alone suffices.
Taking the maximum over entries gives $K$, with
polynomially many bits, such that $A^{Dk+1}\geq0$ for all $k\geq K$; the
Frobenius bound applied to $a=DK+1$ and $b=D(K+1)+1$ then yields a valid
threshold $N\leq ab-a-b+1$ with polynomially many bits.
\end{proof}

\subsection{Lemma~\ref{lem:onepower}, Proposition~\ref{prop:offsetzero}, and Corollary~\ref{cor:existspower}}\label{app:existspower}

\begin{proof}
We first set up the expansion at offset $0$.  On the generalized
$\lambda$-eigenspace the binomial theorem gives
$A^nE_\lambda=\sum_s\binom ns\lambda^{n-s}H_{\lambda,s}$; substitute
$n=Dk$ and group those $\lambda$ with equal $D$th powers to obtain
\begin{equation}\label{eq:grouped-zero}
 (A^{Dk})_{ij}=\sum_C Q_C^{ij,(0)}(k)\mu_C^k,
 \qquad
 Q_C^{ij,(0)}(X)=\sum_{s=0}^{d-1}\binom{DX}{s}b_{C,s}^{ij}(0),
\end{equation}
with
\[
 b_{C,s}^{ij}(0)=\sum_{\lambda\in C}\lambda^{-s}(H_{\lambda,s})_{ij},
\]
valid for every $k$ with $Dk\geq d$: the analogue at offset $0$
of~\eqref{eq:grouped-expansion} and~\eqref{eq:B-definition}.  The bases
$\mu_C$ are unchanged, so modulus comparisons between classes are those
of Step~3 of the algorithm.  The coefficients are as computable as
before: conjugating~\eqref{eq:grouped-zero} and invoking
Lemma~\ref{lem:exppoly}, exactly as in the proof of
Lemma~\ref{lem:grouped}, gives
$\sigma(b_{C,s}^{ij}(0))=b_{\sigma(C),s}^{ij}(0)$ for every
automorphism $\sigma$, so the degree bound, the height bookkeeping,
and the zero and sign tests of Lemma~\ref{lem:separation} apply
verbatim to the $b_{C,s}^{ij}(0)$.  The algorithm
therefore runs Steps~2--4 of Section~\ref{sec:algorithm} with
$b_{C,s}^{ij}(0)$ in place of $b_{C,s}^{ij}$, accepting exactly when
every entry satisfies the criterion of Theorem~\ref{thm:criterion} read
at offset $0$.  It remains to prove that some power of $A$ is
nonnegative if and only if the offset-$0$ criterion holds.

If the criterion holds, the sufficiency direction of
Theorem~\ref{thm:criterion}---a pure asymptotic argument on the
expansion~\eqref{eq:grouped-zero}, unchanged at offset $0$---gives
$A^{Dk}\geq0$ for all large $k$; pick such
a $k$ with $Dk\geq1$; then $n=Dk$ is a nonnegative power.

Conversely, suppose $A^{n_0}\geq0$ for some $n_0\geq1$.  We verify the
offset-$0$ criterion.  Fix
an entry $(i,j)$ and write $u_k=(A^{Dk})_{ij}$.  Only the dominance
condition of the criterion needs checking.  If no class is active in
$u$ then $u$ is eventually zero and there is nothing to check.
Otherwise, pass to the subsequence
\[
 v_m=u_{n_0m}=(A^{Dn_0m})_{ij}=e_i^\top(A^{Dn_0})^me_j.
\]
Since $A^{Dn_0}=(A^{n_0})^D\geq0$, the sequence $v$ is $\Q_+$-rational.
Substituting $k=n_0m$ into the offset-$0$
expansion~\eqref{eq:grouped-zero} gives
$v_m=\sum_C Q_C^{ij,(0)}(n_0m)\,\mu_C^{n_0m}$, so the characteristic
roots of $v$ are the numbers $\mu_C^{n_0}$ for the classes active in
$u$.  The coefficient polynomials remain nonzero after this affine
substitution, so $v$ is not eventually zero by
Lemma~\ref{lem:exppoly}, and the roots remain nondegenerate, since a
root-of-unity ratio $(\mu_C/\mu_{C'})^{n_0}$ would make
$\mu_C/\mu_{C'}$ a root of unity.  Berstel's theorem therefore gives
$v$ a unique dominant root, so $u$ has a unique dominant class $C_*$.
Its leading coefficient is positive: the offset-$0$ expansion gives
$u_k/(k^r\mu_{C_*}^k)\to D^rb_{C_*,r}^{ij}(0)/r!$, and along the
subsequence $k=n_0m$ the left-hand side is nonnegative for every
$m\geq1$ while the limit is nonzero.  Hence the offset-$0$ criterion
holds.
\end{proof}

\section{Two worked examples}\label{app:example}

The exact criterion of Theorem~\ref{thm:criterion} has two branches,
and the algorithm correspondingly rests on two primitives: comparing the
moduli of active classes, and certifying that a grouped coefficient
vanishes or has a definite sign.  The two examples below isolate one
branch each.  The first is an end-to-end yes-instance decided entirely by
dominance---case~(2)---with the one-progression reduction visible; its
eigenvalues are rational, so exact cancellation cannot occur and the
certified tests are never stressed.  The second matrix is nonnegative
outright, so its acceptance is no surprise; its value is that one entry
falls into case~(1) through an exact cancellation of irrational
contributions, the situation that forces the zero test of
Lemma~\ref{lem:separation} to be exact rather than numerical.

\subsection{Case (2): dominance, end to end}

Let
\[
 A_\star=\frac1{30}
 \begin{pmatrix}
 21&-9&21\\
 -9&21&21\\
 21&21&-9
 \end{pmatrix}.
\]
It has negative entries, yet its powers eventually become strictly
positive: with the orthogonal projectors $P_\rho,P_+,P_-$ onto the lines
spanned by $(1,1,1),(1,-1,0),(1,1,-2)$, one has
$A_\star=\frac{11}{10}P_\rho+P_+-P_-$, hence
$A_\star^n=(11/10)^nP_\rho+P_++(-1)^nP_-$.  The smallest entries are
$((11/10)^n-1)/3$ for even $n$ and $((11/10)^n-2)/3$ for odd $n$, so
$A_\star^n>0$ for every $n\geq8$.  The matrix thus exhibits the main
phenomenon: negative entries and a degenerate eigenvalue pair $1,-1$, but
eventually strictly positive powers.

Trace the algorithm.  The eigenvalues are $11/10,1,-1$.  The quotient
$1/(-1)=-1$ has order $2$, while neither quotient involving $11/10$ is a
root of unity (its absolute value is not $1$).  Hence the torsion classes
are $C_\rho=\{11/10\}$ and $C_0=\{1,-1\}$, with $L=2$ and $D=4$; indeed
$(11/10)^4=14641/10000$ and $1^4=(-1)^4=1$, so the degenerate pair
collapses to one positive base.  Since $A_\star$ is diagonalizable, only
$s=0$ occurs, and the grouped expansion is
\[
 A_\star^{4k+1}
 =\frac{11}{10}\left(\frac{14641}{10000}\right)^kP_\rho+(P_+-P_-),
\]
so $B_{C_\rho,0}=(11/10)P_\rho$, with every entry $11/30>0$, and
$B_{C_0,0}=P_+-P_-$.  Comparing the squared moduli $121/100>1$ selects
$C_\rho$ as the unique dominant class for every entry, every leading
coefficient equals $11/30>0$, and the algorithm accepts, in agreement with
the direct formula.  The one-progression reduction is also visible:
$A_\star^{4k+1}\geq0$ for $k\geq2$, the exponents $9$ and $13$ are
coprime, and every $n>9\cdot13-9-13=95$ is a nonnegative combination of
them.  (The direct threshold is $8$; the Frobenius bound is deliberately
coarse.)

\subsection{Case (1): exact cancellation across root fields}

Let
\[
 R=\begin{pmatrix}0&0&2\\1&0&0\\0&1&0\end{pmatrix},
\]
the companion matrix of $X^3-2$, so that $R^3=2I$.  Its eigenvalues
$\sqrt[3]2$, $\omega\sqrt[3]2$, $\overline\omega\sqrt[3]2$, for a
primitive cube root of unity $\omega$, have pairwise ratios of order
$3$, so they form a single torsion class $C$, with $L=3$, $D=6$, and
class root $\mu_C=(\sqrt[3]2)^6=4$.  The three root fields
$\Q(\lambda)$ are distinct cubic fields---one real, two complex
conjugates---inside a splitting field of degree $6$, which the
algorithm never constructs.

Since $R$ is diagonalizable, only $s=0$ occurs, and the projectors are
$E_\lambda=(R^2+\lambda R+\lambda^2I)/(3\lambda^2)$.  The local
contribution of $\lambda$ to an entry is
\[
 T_{\lambda,0}^{ij}
 =\lambda(E_\lambda)_{ij}
 =\tfrac13\bigl((R^2)_{ij}\lambda^{-1}+R_{ij}+\lambda\delta_{ij}\bigr)
 \in\Q(\lambda),
\]
computed entirely within its own field.  In the $(1,1)$ entry the three
contributions are $\lambda/3$ for the three roots: three irrational
numbers in three different fields, whose sum
$(1+\omega+\overline\omega)\sqrt[3]2/3=0$ cancels exactly.  In the
$(1,3)$ entry each contribution is $2/3$ and the sum is $2$.  Indeed
$B_{C,0}=\sum_\lambda\lambda E_\lambda=R$, in agreement with
$R^{6k+1}=(R^3)^{2k}R=4^kR$.

The algorithmic content is case~(1) of Theorem~\ref{thm:criterion}:
since $b_{C,0}^{11}=0$ and $C$ is the only class, no class is
$(1,1)$-active, and the algorithm classifies that entry as eventually
zero along the progression---indeed $(R^{6k+1})_{11}=0$ for every $k$.
The certified test reaches this verdict by summing balls computed
independently in the three embeddings; a numerical test could shrink
the ball around $0$ forever without certifying exactness, and it is the
degree and height bounds of Lemma~\ref{lem:grouped} that convert
``smaller than $2^{-\poly}$'' into ``exactly zero''.  The orbit bound
is sharp here: with a single torsion class the $G$-invariant partition
has one block, so Lemma~\ref{lem:orbit} certifies \emph{in advance}
that every grouped coefficient of $R$ is rational---low degree emerging
from symmetry alone, while every individual summand is irrational.

\end{document}